\documentclass[11pt,a4paper]{amsart}
\usepackage[T1]{fontenc}
\usepackage[latin9]{inputenc}
\usepackage{amssymb}
\usepackage{color}
\usepackage{graphicx,color}
\usepackage{amsmath, amssymb, graphics}

\makeatletter
\numberwithin{equation}{section} 
\numberwithin{figure}{section} 
  \@ifundefined{theoremstyle}{\usepackage{amsthm}}{}
  \theoremstyle{plain}
  \newtheorem{thm}{Theorem}[section]
  \theoremstyle{plain}
  \newtheorem{cor}[thm]{Corollary}
  \theoremstyle{plain}
  \newtheorem{prop}[thm]{Proposition}
  \theoremstyle{remark}
  \newtheorem{rem}[thm]{Remark}
  \theoremstyle{remark}
  
  \theoremstyle{plain}

\def\bfR#1{{\bf R}^#1}

\def\com#1{ \hbox{#1}}

\smallskip
\def\<{{\langle }}
\def\>{{\rangle }}

\def\bfR#1{{\bf R}^#1}

\def\com#1{ \quad\hbox{#1}\quad}

\smallskip
\def\<{{\langle }}
\def\>{{\rangle }}
\def\k{\rm {\bf Kg}}
\def\m{\rm {\bf m}}
\def\s{\rm {\bf s}}
\def\um{\rm {\bf um}}
\def\ud{\rm {\bf ud}}
\def\ut{\rm {\bf ut}}

\begin{document}

\title[Relativistic Lagrange Points]{Existence and Stability the Lagrangian point $L_4$ for the Earth-Sun system under a relativistic framework}
\author{Oscar Perdomo}
\date{\today}

\curraddr{Department of Mathematics\\
Central Connecticut State University\\
New Britain, CT 06050\\}
\email{ perdomoosm@ccsu.edu}


\begin{abstract}

It is well known that, from the Newtonian point of view, the Lagrangian point $L_4$ in the circular restricted three body is stable if $\mu< \frac{1}{18}(9-\sqrt{19})\approx 0.03852$. In this paper we will provide a formula that allows us to compute the eigenvalues of the matrix that determines the stability of the equilibrium points of a family of ordinary differential equations. As an application we will show that, under the relativistic framework, the Lagrangian point $L_4$ is also stable for the Sun-Earth system.  Similar arguments show the stability for $L_4$ not only for the Sun-Earth system but  for systems coming from a range of values for $\mu$ similar to those in the Newtonian restricted three body problem.

\end{abstract}

\maketitle

\section{Introduction}

Let us consider a circular solution of the two body problem. Let us assume that these two bodies have masses $m_1$ and $m_2$ and let us call them primaries bodies. The study of the motion of a third body with mass $m_3$ that moves under the effect of the gravitational force produced by the primary bodies, under the condition that the mass $m_3$ is so small that it does not  affect the motion of the primaries, is known as the circular restricted three body problem (Circular RTBP). We assume that the $m_2\le m_1$ and we define $\mu=\frac{m_2}{m_1+m_2}\le 1/2$. Lagrangian points can be viewed as possible motions with the property that the distance of the body with mass $m_3$ to the primaries is constant. There are 5 possible  motions and they are denoted by $L_1$, $L_2$, $L_3$, $L_4$ and $L_5$. For the motions $L_1$, $L_2$, $L_3$ the three bodies stay in the same line. For the motions $L_4$ and $L_5$, the three bodies alway lie on the vertices of an equilateral triangle.

After doing a change of coordinates by considering a rotating frame (see section \ref{units}), we can see that the Lagrangian points are the equilibrium points of an ordinary differential equation and they are found by solving a system of two equations with two variables. Even though this system is not trivial, we can find formulas for the solution, this is, we have exact expressions for the solutions. Having this exact formula for the Lagrangian points allows to show that the points $L_1$, $L_2$, $L_3$ are always unstable and $L_4$ and $L_5$ are stable for values of $\mu<\frac{1}{18}(9-\sqrt{19})$.

When we consider the problem under the relativistic point of view the situation is very similar in nature, in the sense that the Lagrangian points are the equilibrium points of an ordinary differential equation and they are found by solving a system of two equations with two variables. The difference is that the solution of the system may not have an exact analytical expression, even the existence of the solution can be questioned mathematically. So far the Lagrange points under the relativistic point of view have been studied by first approximating the solution under the assumption that this solution exists, and the stability of this equilibrium points have been studied by approximating the characteristic polynomial that produces the eigenvalues that determines the  stability of the equilibrium point. For this reason it is not surprising that some papers states that the relativistic equilibrium point $L_4$ is unstable for all values of $\mu$ (\cite{BH}) and others (\cite{A} and \cite{D}) states that this point is stable for values of $\mu$ close to those obtained in the Newtonian case.

For the reasons explained above we have that in order to mathematically establish the stability of the Lagrangian points under relativistic framework we need to solve the following questions/difficulties:
\begin{enumerate}
\item
Can we show that the system of two equations with two variables that defines the Lagrangian points has solutions? These equations are displayed on page 273 of the paper \cite{BH}. And, if there are solutions, can we find a closed expression for them?
\item
The stability of $L_4$ in the Newtonian case is proven by showing that the eigenvalues of a 4 by 4 matrix are complex numbers with zero real part. In order to prove the stability of $L_4$ in the relativistic case we need to prove that these eigenvalues also have real part equal to zero. The set of complex numbers with zero real part (a line) is a set with measure zero and therefore no rounding can be done to show that a complex number has zero real part.  Here is the second question: how to prove, without having an exact expression for the equilibrium points, that these eigenvalues in the relativistic case have zero real part?

\end{enumerate}

The main contribution of this paper is related with the second numeral presented above. We provide  an exact analytical expression for the characteristic polynomial of the matrix that determines the stability of $L_4$, see Theorem \ref{cp}. This analytic expression is obtained by assuming that the matrix is evaluated at the exact solution of the system of equations that determines $L_4$.

The second  result in our paper has to do with the numeral one  presented above. In section \ref{ex} we will prove the existence of the point $L_4$ for the Earth-Sun system using the Poincare-Miranda Theorem. Regarding this section the author wants to apologize for using quotient of integers instead of decimals. The reason for doing so is that software like Wolfram Mathematica can work with infinity precision (no roundoff error) when rational numbers are used instead of decimals. A decimal approximation will be placed in front of every quotient of two integers with the intension of giving an idea of what each number is located in the real line.

The intermediate value theorem states that a continuous function $f(x)$ defined on a closed interval $[a,b]$ must have a zero under the assumption that $f(a)<0$ and $f(b)>0$. The Poincare-Miranda Theorem is a generalization of the intermediate value theorem. We will be using the two dimensional version of the Poincare-Miranda theorem that states that if two continuous functions on two variables $f(x,y)$ and $g(x,y)$ are defined on the rectangle $[a,b] \times [c,d]$ and  $f(x,c)>0$ and $f(x,d)<0$ for all $x\in[a,b]$ and $g(a,y)>0$ and $g(b,y)<0$ for all $y\in[c,d]$ then, there is point a $(x_0,y_0)$ inside the rectangle $[a,b]\times[c,d]$ that satisfies both equations: $f(x_0,y_0)=0$ and $g(x_0,y_0)=0$. In other words, the theorem states that if the function $f$ is negative in the lower side of the rectangle and positive in the upper side of the rectangle, and moreover, the function $g$ is negative in the left side of the rectangle and positive in the right side of the rectangle, then, there must be a point inside the rectangle that solves the system $f=0$ and $g=0$.

With the intension to set up notation for the change of units and to get familiar with the terminology, section 2 explains the Newtonian case. Section 3 displays the ODE for the relativistic case. Section 4 proves mathematically the existence of $L_4$ for the relativistic RTBP  for the Sun-Earth system. Section 5 shows the Theorem that allows us to compute without rounding the characteristic polynomial that provides the stability of $L_4$. An important aspect of this formula is that it mathematically shows that this polynomial has the form $\lambda^4+a_1 \lambda^2+a_2$. This is not obvious. For example, the reason the paper \cite{BH}  shows the instability of the point $L_4$ for all the range of $\mu$ is because the method used to obtain the characteristic  polynomial leads the authors to an expression of the form $b_1(c,\mu) \lambda^4+
b_2(c,\mu) \lambda^3+b_3(c,\mu) \lambda^2+b_4(c,\mu) \lambda+b_5(c,\mu)$ with $b_2$ and $b_4$ different form zero. See \cite{BH} page 277.
 
The author would like to thank Charles Simo, David R. Skillman and Andr\'es Mauricio Rivera for their valuable comments.

\section{Circular solutions of the two body problem and changing units} \label{units} Let us consider two bodies, (the primaries), with masses $m_1\,  \k$ and $m_2\,  \k$ moving in the space with positions $x$ and $y$. Let us take the gravitational constant to be equal to $G=6.67384*10^{-11}\, \m^3\, \k^{-1}\, \s^{-2}$. It is easy to check that for a given positive number $a$, the functions

$$
x(t)=\frac{-m_2\, a}{m_1+m_2} \, (\cos(\omega\,  t),\sin(\omega \, t)), \quad y(t)=\frac{m_1\, a}{m_1+m_2} \, (\cos(\omega\,  t),\sin(\omega \, t))
$$

with $\omega=\sqrt{\frac{G(m_1+m_2)}{a^3}}\, {\bf s}^{-1}$ satisfy the two body problem ODE 

$$\ddot{x}=\frac{m_2 G}{|x-y|^3} \, (y-x)\quad \ddot{y}=\frac{m_1 G}{|y-x|^3} \, (x-y)$$

This solution satisfies that the distance between the two bodies is always $a$ meters and moreover, both motions are periodic since they complete a revolution after $T=\frac{2 \pi}{\omega}=2 \pi\, \sqrt{\frac{a^3}{G(m_1+m_2)}}\, {\bf s} $. Let us change the units of mass, distance and time in the following way: Let us denote by $\um$ the unit of mass such that $1 {\um} = (m_1+m_2)\,  \k$, let us denote by $\ud$ the unit of distance such that $1\,  {\ud}=a\,  \m$ and  finally let us denote by $\ut$,  the unit of time $\ut$ such that 
$1\, {\ut} = \sqrt{\frac{a^3}{G (m_1+m_2)}}\quad \s$. Notice that using the units $\ut$ and $\ud$ we have that the distance between the two bodies is $1 \ud$ and the period of the motion is $2 \pi \, \ut$. We also have that the gravitation constant is $1\, 
\ud^3\, \um^{-1}\, \ut^{-2}$. We point out that the speed of light is

\begin{eqnarray}\label{sl}
 c=299792458 * \sqrt{\frac{a}{G (m_1+m_2)}}\, \frac{\ud}{\ut}
 \end{eqnarray}

If we denote by  $\mu=\frac{m_2}{m_1+m_2}$ and we work in the new units $\ut$, $\um$, $\ud$, then the mass of the first and  second body are $1-\mu$ and $\mu$ and  the motion of the primaries are given by 

$$ x(t)=-\mu \left(\cos (t),\sin (t)\right) \quad  y(t)=(1-\mu ) \, \left(\cos (t),\sin (t)\right) $$

Moreover, if a third body with position $z(t)$ and neglecting mass compared with $m_1$ and $m_2$ moves under the influence of the gravitational force of the primaries, then $z$ satisfies

\begin{eqnarray}\label{eq1}
\ddot{z}=\frac{(1-\mu)}{|x-z|^3}\, (x-z)+\frac{\mu}{|y-z|^3}\, (y-z)
\end{eqnarray}

A direct computation shows that if we take 

\begin{eqnarray}\label{eq2}
z = \left(\xi (t) \cos (t)-\eta (t) \sin (t),\eta (t) \cos (t)+\xi (t) \sin (t)\right)\, ,
\end{eqnarray}

then, (\ref{eq1}) reduces to 
\begin{eqnarray}\label{eq3}
\ddot{\xi}-2 \dot{\eta}=\frac{\partial w_0}{\partial \xi}
\com{and} 
\ddot{\eta}+2 \dot{\xi}=\frac{\partial w_0}{\partial \eta}
\end{eqnarray}

where,

\begin{eqnarray}\label{eq4}
w_0 = \frac{1}{2} (\xi^2+\eta^2)+\frac{1-\mu}{\sqrt{(\xi+\mu)^2+\eta^2}}+\frac{\mu}{\sqrt{(\xi+\mu-1)^2+\eta^2}}
\end{eqnarray} 

A direct verification shows that $\xi(t)=\frac{1-2 \mu}{2}$ and $\eta(t)=\frac{\sqrt{3}}{2}$ is a solution of the system of equations (\ref{eq3}). This equilibrium point $(\frac{1-2 \mu}{2},\frac{\sqrt{3}}{2})$ is known as the Lagrangian point $L_4$. In order to analyze the stability of $L_4$ we consider the function 

$$F_0=(\dot{\xi} ,2 \dot{\eta}+\frac{\partial w_0}{\partial \xi },\dot{\eta} ,\frac{\partial w_0}{\partial \eta }-2 \dot{\xi} ) $$

as a function of the variables $\phi=(\xi,\dot{\xi},\eta,\dot{\eta})$. It is easy to check that the ODE (\ref{eq3}) is equivalent to the ODE
$\dot{\phi}=F_0(\phi)$. In order to analyze the stability of the equilibrium solution $\phi_0=(\frac{1-2 \mu}{2},0,\frac{\sqrt{3}}{2},0)$, we compute the 4 by 4 matrix $A_0=DF_0$ evaluated at $  \xi=\frac{1-2 \mu}{2},\, \eta=\frac{\sqrt{3}}{2},\, \dot{\xi}=0, \, \dot{\eta}=0$. Since we can check that the characteristic polynomial of the matrix $A_0$ is equal to

$$\lambda ^4+\lambda ^2-\frac{27}{4} (\mu -1) \mu$$

Then, we conclude that, when either $0<\mu<\frac{1}{18} \left(9-\sqrt{69}\right)$ or $\frac{1}{18} \left(9+\sqrt{69}\right)<\mu<1$, then the real part of all the eigenvalues of $A_0$ is zero and therefore $L_4$ is linearly stable. It is known that there are 5 equilibrium solutions for the ODE (\ref{eq3}); we have $L_4$, given above, $L_5=(\frac{1-2 \mu}{2},-\frac{\sqrt{3}}{2})$ and three more of the form $(\xi_1,0)$, $(\xi_2,0)$ and $(\xi_3,0)$ usually labeled as the Lagrangian points $L_1$, $L_2$ and $L_3$. A similar analysis to the one that we just did for $L_4$ can be done for the other Lagrange points to conclude that $L_5$ is also linear stable for the same range of the parameter $\mu$ and, $L_1$, $L_2$ and $L_3$ are linearly unstable.

\section{The ODE in the relativistic case:} \label{EQ} In the relativistic case, the equation of the motion for the restricted three body problem are very similar to the one given by Equation (\ref{eq3}). It takes the form (see Brumberg, 1972, \cite{B} and Bhatnagar \cite{BH})

\begin{eqnarray}\label{eq6}
\ddot{\xi}-2 n \dot{\eta}=\frac{\partial w}{\partial \xi}-\frac{d}{dt} (\frac{\partial w}{\partial \dot{\xi}})\com{and} 
\ddot{\eta}+2 n \dot{\xi}=\frac{\partial w}{\partial \eta}-\frac{d}{dt} (\frac{\partial w}{\partial \dot{\eta}})
\end{eqnarray}
 
 where $w=w_0+\frac{1}{c^2}\, w_1$ with
 
\begin{eqnarray*}
w_1 &=&-\frac{3}{2} \left(1-\frac{1}{3} \mu  (1-\mu )\right) {\rho}^2+\frac{1}{8} \left(\dot{\eta}^2+2 (\dot{\eta}\xi -\dot{\xi}\eta )+\dot{\xi}^2+{\rho}^2\right)^2+\\
& & \frac{3}{2} \left(\frac{1-\mu }{{\rho_1}}+\frac{\mu }{{\rho_2}}\right) \left(\dot{\eta}^2+2 (\dot{\eta}\xi -\dot{\xi}\eta )+\dot{\xi}^2+{\rho}^2\right)-\frac{1}{2} \left(\frac{(1-\mu )^2}{{\rho_1}^2}+\frac{\mu ^2}{{\rho_2}^2}\right)+\\
& & \mu (1-\mu)\left(\left(4 \dot{\eta}+\frac{7 \xi }{2}\right) \left(\frac{1}{{\rho_1}}-\frac{1}{{\rho_2}}\right)-\frac{1}{2} \eta ^2 \left(\frac{\mu }{{\rho_1}^3}+\frac{1-\mu }{{\rho_2}^3}\right)+\left(\frac{3 \mu -2}{2 {\rho_1}}-\frac{1}{{\rho_1} {\rho_2}}+\frac{1-3 \mu }{2 {\rho_2}}\right)\right.\\
n&=&1-\frac{3}{2 c^2} \left(1-\frac{1}{3} \mu(1-\mu) \right)\\
\rho&=& \sqrt{\xi^2+\eta^2},\quad \rho_1=\sqrt{(\xi+\eta)^2+\eta^2}\com{and} \rho_2=\sqrt{(\xi+\eta-1)^2+\eta^2}
\end{eqnarray*}

\subsection{The system of equations}\label{eqns}

The equilibrium points of the system of differential equations given by (\ref{eq6}) are the solutions of the system $f=0$ and $g=0$ where  $f=\frac{\partial w}{\partial \xi }$ and $g=\frac{\partial w}{\partial \eta }$ evaluate at $\dot{\xi}=\dot{\eta}=0$. These equations are display on page  273 of the paper \cite{BH}.

\section{Existence  of $L_4$ in the relativistic case using the Poincare-Miranda theorem}\label{ex}

\subsection{Existence of $L_4$ for the Earth-Sun system}\label{ESsystem}

For computation in this section, we will assume that the earth and the sun are moving with a constant distance between them of $a_0=149597870700$ {\bf m}, with  the mass of the sun equal to $M_0=1.988544* 10^{30}\, {\bf Kg}$ and the mass of the earth equal to $5.9729 *10^{24} {\bf Kg}$. We will also be taking the speed of light to be $c_0=299792458 \, \frac{{\bf m}}{{\bf s}}$. Using this data we have that the values for $\mu$ and $c$ are given by 
$$\mu= \frac{59729}{19885499729}\approx 3.00365*10^{-6}$$

and

$$ c=c_0 * \sqrt{\frac{a_0}{G (m_0+M_0)}}=\frac{149896229 \sqrt{\frac{1495978707}{3317816087784734}}}{10}\approx 10065.3$$

We will prove the existence of the relativistic point $L_4$ for the Earth-Sun system. This is, we will prove the existence of a point $(\xi_0,\eta_0)$ that is within a close to the point $(\frac{1-2 \mu}{2},\frac{\sqrt{3}}{2})$ that satisfies the equation 
\begin{eqnarray}\label{the eq}
f(\xi_0,\eta_0)=0\com{and} g(\xi_0,\eta_0)=0
\end{eqnarray} 

In order to prove the existence of $(\xi_0,\eta_0)$, let us consider the following five points

\begin{eqnarray*}
P_1&=& \left(\frac{312498189077}{625000000000},\frac{4330127145451}{5000000000000}\right)\\
&=&\left(0.4999971025232, 
0.8660254290902\right)\\
P_2&=&\left(\frac{312498095327}{625000000000},\frac{4330127145451}{5000000000000}\right) \\
&=&\left(0.4999969525232, 
0.8660254290902\right)\\
P_3&=&\left(\frac{312498189077}{625000000000},\frac{43301267124383}{50000000000000}\right) \\
&=&\left(0.4999971025232, 
0.86602534248766\right)\\
P_4&=&\left(\frac{312498095327}{625000000000},\frac{43301267124383}{50000000000000}\right)\\
&=&\left(0.4999969525232, 
0.86602534248766\right)\\
\end{eqnarray*}

And let $\beta_1$ be the line that connects $P_2$ with $P_1$, $\beta_2$ be the line that connect $P_4$ with $P_3$, $\beta_3$ be the line that connect $P_3$ with $P_1$ and $\beta_4$ be the line that connect $P_4$ with $P_2$. More precisely,
\begin{eqnarray*}
\beta_1(t)&=&t P_1+(1-t) P_2\\
\beta_2(t)&=&t P_3+(1-t) P_4\\
\beta_3(t)&=&t P_1+(1-t) P_3\\
\beta_4(t)&=&t P_2+(1-t) P_4\\
\end{eqnarray*}

\begin{figure}[hbtp]
\begin{center}\includegraphics[width=.7\textwidth]{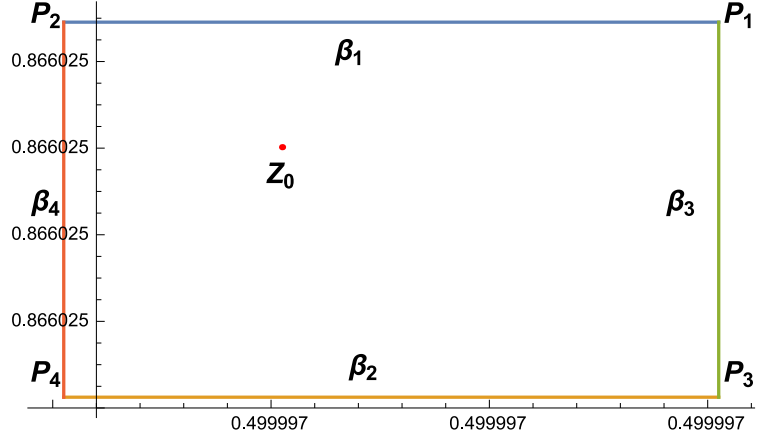}
\end{center}
\caption{The Poincare Miranda theorem guarantees that the system of equations $f=0$ and $g=0$ has a solution inside the region above }\label{fig1}
\end{figure}
\begin{thm}
There exists a solution of the system of equations  $f=0$ and $g=0$ inside the rectangle delimited by the union of the curves $\beta_1$, $\beta_2$, $\beta_3$ and $\beta_4$. Moreover, we have that the first six significant digits of the functions evaluated at the points $Z_0$, $P_1$, $P_2$, $P_3$ and $P_4$ are given by 

\begin{eqnarray*} 
f(P_1)&=& 1.124997\dots * 10^{-7} \qquad g(P_1)\, =\,  1.94854\dots * 10^{-7}\\
f(P_2)&=& -2.22877\dots * 10^{-13} \qquad g(P_2)\, =\,  3.84116\dots * 10^{-13}\\
f(P_3)&=& 4.56885\dots * 10^{-13}  \qquad g(P_3)\, =\, -7.69391\dots * 10^{-13}\\
f(P_4)&=& -1.12499\dots * 10^{-7} \qquad g(P_4)\, =\, - 1.94855\dots * 10^{-7} \\
\end{eqnarray*}

We also have that $g>0$ on $\beta_1$, $g<0$ on $\beta_2$, $f>0$ on $\beta_3$ and $f<0$ on $\beta_4$. As a consequence of the Poincare-Miranda Theorem we conclude that there exists a point $P_0=( \xi_0,\eta_0)$ inside the region bounded by the four curves $\beta_i$ such that  $f(P_0)=g(P_0)=0$.
\end{thm}

\begin{proof}
The proof to the theorem relies on the fact that we have an exact expression (an analytic expression) for $f$ and $g$ and the fact that programs like Mathematica allow us to precisely compute any desired amount of real digits of an exact expression. In order to obtain these digits, it is required that we work with exact numbers, that is the reason we decided to use rational numbers and not decimals. The computation for the values of the functions $f$ and $g$ were obtained by using the command $\text{RealDigits}[f(P_i),10,6]$ and $\text{RealDigits}[g(P_i),10,6]$ form the program Wolfram Mathematica 10. It is not difficult to check that the directional derivative of $g$ along the velocity of the curves $\beta_1$ and $\beta_2$ does not change sign, they both are approximately $1.29903$, and also, the directional derivative of $f$ along the curves $\beta_3$ and $\beta_4$ does not change sign, they both are  approximately $1.29903$. Then, we conclude that the function $g$ is monotonic along $\beta_1$ and $\beta_2$, this fact along with the values of $g$ at the endpoints allows to prove that $g$ is positive on $\beta_1$ and negative on $\beta_2$. A similar arguments holds for the function $f$. This concludes the proof  of the  theorem.
\end{proof}

\section{Exact expression for the characteristic polynomial at the equilibrium points }\label{cps}

The following theorem provides an expression for the characteristic polynomial of a system of the type given by the relativistic three body problem.

\begin{thm}\label{cp}
Let us consider the potential $U=U(x,\dot{x},y,\dot{y})$ and let us consider the following system of ODE

$$\ddot{x}-2 k \dot{y}=\frac{\partial U}{\partial x}-\frac{d\, }{dt}(\frac{\partial U}{\partial \dot{x}}),  \quad
\ddot{y}+2 k \dot{x}=\frac{\partial U}{\partial y}-\frac{d\, }{dt}(\frac{\partial U}{\partial \dot{y}})$$
\end{thm}
where $k$ is a constant. If $L_0=(x_0,y_0)$ is an equilibrium point of the system above and 
$d= 1+\frac{\partial^2U}{\partial \dot{x}^2}+\frac{\partial^2U}{\partial \dot{y}^2}+\frac{\partial^2U}{\partial \dot{x}^2} \frac{\partial^2U}{\partial \dot{y}^2}-\left(\frac{\partial^2U}{\partial \dot{x} \partial \dot{y}}\right)^2$ is not zero at $\tilde{L}_0=(x_0,0,y_0,0)$, then, the characteristic polynomial of the matrix that describes the linearization of the ODE at $L_0$ is given by

$$ \lambda^4 + a_1\lambda^2+a_2$$

where,
\begin{eqnarray*}
a_2\, d&=&\frac{\partial^2U}{\partial y^2} \frac{\partial^2U}{\partial x^2}-\left(\frac{\partial^2U}{\partial x \partial y}\right)^2
\end{eqnarray*}

and 

\begin{eqnarray*}
a_1\, d&=&-4 k \frac{\partial^2U}{\partial y \partial \dot{x}}+4 k \frac{\partial^2U}{\partial x \partial \dot{y}}+\left(\frac{\partial^2U}{\partial y \partial \dot{x}}\right)^2+\left(\frac{\partial^2U}{\partial x \partial \dot{y}}\right)^2-\frac{\partial^2U}{\partial y^2}-\frac{\partial^2U}{\partial y^2} \frac{\partial^2U}{\partial \dot{x}^2}-\\
& &2 \frac{\partial^2U}{\partial y \partial \dot{x}} \frac{\partial^2U}{\partial x \partial \dot{y}}+2 \frac{\partial^2U}{\partial \dot{x} \partial \dot{y}} \frac{\partial^2U}{\partial x \partial y}-\frac{\partial^2U}{\partial \dot{y}^2} \frac{\partial^2U}{\partial x^2}-\frac{\partial^2U}{\partial x^2}+4 k^2
\end{eqnarray*}

\begin{proof}
The ODE in this theorem can be reducee to the first order ODE $\dot{\phi}=F(\phi)$ with $\phi=(x,\dot{x},y,\dot{y})$ and 
$$ F(\phi)=(\dot{x},{F_2}(\phi),\dot{y},{F_4}(\phi))$$
where the functions ${F_2}$ and ${F_4}$ are given as the solution, near $\tilde{L}_0=(x_0,0,y_0,0)$, of the system of equations

\begin{eqnarray*}
{F_2}-2k\dot{y}&=&\frac{\partial U}{\partial x}-\dot{x} \frac{\partial^2U}{\partial \dot{x} \partial x}-{F_2} \frac{\partial^2U}{\partial \dot{x}^2}-\dot{y} \frac{\partial^2U}{\partial \dot{x} \partial y}-{F_4} \frac{\partial^2U}{\partial \dot{x} \partial \dot{y}}\\
{F_4}+2k\dot{x}&=&\frac{\partial U}{\partial y}-\dot{x} \frac{\partial^2U}{\partial \dot{y} \partial x}-{F_2} \frac{\partial^2U}{\partial \dot{x} \partial \dot{y} }-\dot{y} \frac{\partial^2U}{\partial \dot{y} \partial y}-{F_4} \frac{\partial^2U}{\partial \dot{y}^2}
\end{eqnarray*}

Recall that we have that ${F_2}(\tilde{L}_0)={F_4}(\tilde{L}_0)=0$. If we compute the partial derivative with respect to $x$ to the system of equations above and we evaluate at $\tilde{L}_0$, we get the following system of equation 

\begin{eqnarray*}
{F_2}_x&=&\frac{\partial^2 U}{\partial x^2}-{F_2}_x \frac{\partial^2U}{\partial \dot{x}^2}-{F_4}_x \frac{\partial^2U}{\partial \dot{x} \partial \dot{y}}\\
{F_4}_x&=&\frac{\partial^2 U}{\partial x\partial y}-{F_2}_x \frac{\partial^2U}{\partial \dot{x} \partial \dot{y} }-{F_4}_x \frac{\partial^2U}{\partial \dot{y}^2}
\end{eqnarray*}

This is a linear system on ${F_2}_x$ and  ${F_4}_x$ with solution solution satisfying,

\begin{eqnarray*}
{F_2}_x\, d&=&\frac{\partial^2U}{\partial \dot{y}^2} \frac{\partial^2U}{\partial x^2}+\frac{\partial^2U}{\partial x^2}-\frac{\partial^2U}{\partial \dot{x} \partial \dot{y}} \frac{\partial^2U}{\partial x \partial y}\\
{F_4}_x\, d&=& \frac{\partial^2U}{\partial \dot{x}^2} \frac{\partial^2U}{\partial x \partial y}+\frac{\partial^2U}{\partial x \partial y}    -
\frac{\partial^2U}{\partial \dot{x} \partial \dot{y}} \frac{\partial^2U}{\partial x^2}
\end{eqnarray*}

Likewise we can obtain expression for ${F_2}_y$ and ${F_4}_y$ and for ${F_2}_{\dot{x}}$, $g_{\dot{x}}$ and finally for ${F_2}_{\dot{y}}$, ${F_4}_{\dot{y}}$ evaluated at the point $\tilde{L}_0$. The theorem follow after replacing these expression for the partial derivative of the functions ${F_2}$ and ${F_4}$ into the characteristic polynomial of the matrix

$$
\left(
\begin{array}{cccc}
 0 & 1 & 0 & 0 \\
 {F_2}_x & {F_2}_{\dot{x}} & {F_2}_y & {F_2}_{\dot{y}} \\
 0 & 0 & 0 & 1 \\
 {F_4}_x & {F_4}_{\dot{x}} & {F_4}_y & {F_4}_{\dot{y}} \\
\end{array}
\right)
$$

\end{proof}

\begin{rem} When we replace $U$ with $\omega_0$, given in section \ref{units}, and $k=1$, $x_0=\frac{1-2\mu}{2}$ and $y_0=\frac{\sqrt{3}}{2}$, then as it is expected, we obtain that $a_1=1$ and $a_2=-\frac{27}{4}(\mu-1)\mu$
\end{rem}

\begin{rem} We will not be using the well-known expressions  $x_0=\frac{1-2 \mu }{2} (\frac{5}{4 c^2}+1)$ and $y_0=\frac{\sqrt{3}}{2}\, (1-\frac{6 \mu ^2-6 \mu +5}{12 c^2})$ for the relativistic point $L_4$, (see \cite{BH}), due to the fact they only provides an approximation of the coordinates for $L_4$ and as we pointed out before, not approximation will provide a  proof that the eigenvalues have norm exactly 1.
\end{rem}

Once we have shown that the coefficients for $\lambda$ and $\lambda^3$ are zero at the equilibrium point is is somehow easier to check the stability. The next proposition 
\begin{prop}\label{prop} For the differential equation defined on Theorem (\ref{cp}), we have that, if an equilibrium point $(x_0,y_0)$  lies on a region $\Omega$ in $\bfR{2}$ and for every $(x,y)\in \Omega$,  

$$a_1>0,\quad a_2>0\com{and} a_1^2>4a_2$$

then, this equilibrium point $(x_0,y_0)$ is stable.
\end{prop}

\begin{proof} Since the eigenvalues at the equilibrium point $(x_0,y_0)$ satisfies the equation $\lambda^4+a_1\lambda^2+a_2=0$ then, we obtain that either

$$\lambda^2=-\frac{a_1}{2}+\frac{1}{2}\sqrt{a_1^2-4a_2}\com{or}\lambda^2=-\frac{a_1}{2}-\frac{1}{2}\sqrt{a_1^2-4a_2}$$ 

Then, it follows, in either case, $\lambda^2$ is a real negative number. Therefore the real part of the eigenvalues are zero and the proposition follows.
\end{proof}

\begin{cor} In the Earth-Sun system (see section \ref{ESsystem}) the Lagrangian point $L_4$ is stable.
\end{cor}

\begin{proof}
The proof follows by checking the inequalities in Proposition \ref{prop} for the expression for $a_1$ and $a_2$ with domain in the rectangle displayed in Figure \ref{fig1}. In these expression we have replaced $c$ by $\frac{149896229 \sqrt{\frac{1495978707}{3317816087784734}}}{10}$ and $\mu$ by $\frac{59729}{19885499729}$.  In order to prove these inequalities we have used the method of Lagrange multipliers. To make the computation easier we first find bound for the expression $d 
$ defined in Theorem \ref{cp}. In this case a direct verification shows that $d$ lies between $1.000000078898$ and $1.000000079175$. The bounds above allow us to easily obtain bounds for $a_1$ and $a_2$ once we obtain bounds for the functions $a_1 d$ and $a_2 d$. The expressions for $a_1d$ and $a_2 d$ are explicitly given in Theorem \ref{cp}. Both expressions are very long. In order to get the desired bounds, we expanded each one of them. It is not difficult to show that they can be written as linear combinations of expressions of the form

\begin{eqnarray}\label{basis}
 x^{n_1} y^{n_2} \left(\sqrt{(\mu +x-1)^2+y^2} \right)^{-n_3} \left(\sqrt{(\mu +x)^2+y^2}\right)^{-n_4}
 \end{eqnarray}

where $n_1$, $n_3$, $n_3$ and $n_4$ are nonnegative integers. Using a computer software (The command {\it Expand[ ]} from Mathematica) we obtain that 

$$a_1=\sum_{i=1}^{686} a_{1i}\com{and} a_2=\sum_{i=1}^{2227} a_{2i} $$

where each $a_{ji}$ is the product of a number with an expression of the form given in Equation (\ref{basis}). The next step is to use the method of Lagrange Multipliers to each one of the expressions $a_{ji}$. This is also done with the help of the computer, this time we use the command {\it Maximize[ ]} and {\it Minimize[ ]} from Mathematica. By adding all the lower bounds and all the upper bounds we obtain that 

$$ 0.9999986963615272< a_1d < 1.0000012838965228$$
and

$$ 0.0000171807467861<  a_2 d<0.0000233683521331 $$

The bounds above along with the bounds for $d$ easy allow us to conclude the stability of the Lagrangian point $L_4$ for the Sun-Earth system.

\end{proof}

\begin{rem}
As we can see the values for $a_2$ are very small. Also it can be shown that $a_2 d$ evaluated at $(x,y)=(0.499997,0.866025)$ is positive, it is about $0.0000171438$ while, $a_2 d$ evaluated at $(x,y)=(0.499997,0.866)$ is negative, it is about $-0.0000177725$. As a consequence we needed to be sharp with the bounds. Recall that if for a given equilibrium point, the value for  $a_2 d$ is not positive then this equilibrium point is not stable.
\end{rem}
\section{conclusion}

We have used the Poincare-Miranda Theorem, the Lagrange Multiplier method, the help of a computer to be able to handle big expression and Theorem \ref{cp}, to mathematically prove the stability of the point $L_4$  for the Earth-Sun system in the relativistic circular restrict three body problem.



\end{document}